\documentclass[11pt]{article}

\usepackage{times}
\usepackage{natbib}
\usepackage[margin=1in]{geometry}

\usepackage{enumerate}
\usepackage{graphicx}
\usepackage{microtype}
\usepackage{amsfonts,amsmath,amssymb,amsthm}
\usepackage{algorithm}
\usepackage{algpseudocode}
\usepackage{thmtools}
\usepackage{thm-restate}
\usepackage{hyperref}
\usepackage{cleveref}
\usepackage{booktabs}

\newtheorem{theorem}{Theorem}[section]
\newtheorem{lemma}[theorem]{Lemma}

\newtheorem{definition}[theorem]{Definition}

\theoremstyle{remark}

\newcommand{\EMD}{\mathrm{EMD}}

\newcommand{\dist}{\mathrm{dist}}

\newcommand{\Wone}{W_1^T}

\newcommand{\here}{\mathsf{here}}
\newcommand{\boundary}{\partial}

\usepackage[colorinlistoftodos,textsize=small,color=red!25!white,obeyFinal]{todonotes}

\title{Tree Search With Distributional Predictions}

\author{
Michael Dinitz \thanks{Funded in part by NSF award 2228995.}\\
Department of Computer Science\\
Johns Hopkins University\\
Baltimore, MD 21218\\
\texttt{mdinitz@cs.jhu.edu}
\and
Bob Dong \\
Department of Computer Science\\
Johns Hopkins University\\
Baltimore, MD 21218\\
\texttt{bdong9@jh.edu}
}

\date{}

\begin{document}

\maketitle

\begin{abstract}
Learning-augmented algorithms use machine-learned predictions to improve classical algorithmic guarantees when the predictions are accurate, while retaining rigorous performance guarantees when they are not. We study this paradigm for search on trees. Given a tree $T$ containing an unknown target vertex $t$, an algorithm may query any vertex $v$ and learn which neighbor of $v$ lies on the unique path from $v$ to $t$. The goal is to find $t$ using as few queries as possible.

We consider the distributional setting, in which the target is drawn from an unknown distribution $p$ and the algorithm is given a predicted distribution $\widehat p$ of unknown quality. We give an algorithm with expected query complexity $O\left(H(p)+k\log \eta \right)$, where $H(p)$ is the Shannon entropy of the true distribution and $\eta$ is the earth mover's distance between $p$ and $\widehat p$ in the tree metric. We also provide a matching lower bound that shows our algorithm is asymptotically tight.

Finally, experiments on real-world and synthetic trees show that our prediction-based algorithm can use substantially fewer queries than a simple baseline that trusts the prediction completely.

\end{abstract}

\section{Introduction} \label{sec:introduction}
Learning-augmented algorithms, also known as algorithms with predictions, aim to combine the rigorous guarantees of classical algorithms with predictions that may be obtained from historical data or machine-learning systems. Rather than assuming that a prediction is always correct, the goal is to design algorithms whose performance improves when the prediction is accurate, degrades smoothly as the prediction becomes less accurate, and still admits meaningful guarantees when the prediction is poor.

The simplest example is searching in a sorted array.  Without any
additional information, binary search requires $O(\log n)$ queries. Suppose
instead that we are given a predicted location $\widehat{\alpha}(t)$ for the
true location $\alpha(t)$ of the target. If
$
    \eta = |\alpha(t)-\widehat{\alpha}(t)|
$
is the prediction error, then the classical ``doubling binary search'' 
algorithm by \cite{BENTLEY197682} uses only $O(\log \eta)$ queries. Thus an accurate prediction can lead to a substantial improvement, while the query complexity degrades smoothly as the prediction becomes less accurate. We also note that the error term is bounded above by $n$, so this algorithm is never asymptotically worse than the binary search algorithm that simply ignores the prediction

It is easy to see that a sorted array can equivalently be viewed as a path, where querying a vertex reveals the direction of the target. This viewpoint naturally suggests a generalization from paths to trees, leading to the \emph{Search on Trees} problem: there is an unknown target vertex $t$ in a tree $T$, and querying a vertex $v$ reveals the unique neighbor of $v$ lying on the path from $v$ to $t$. Recent work by \citet{dinitz2026treesearchpredictions} studied the point-prediction version of this
problem, where the algorithm is given a predicted vertex $s$. They showed that the
$O(\log \dist(s,t))$ guarantee for paths cannot be obtained on arbitrary trees. However, for a tree of pathwidth $k$, they gave an algorithm with query complexity
$O(k\log \dist(s,t))$. They also constructed a family of hard instances showing that this dependence is asymptotically tight.

However, point predictions are often less natural in machine learning settings. In many applications, our prediction is either generated by looking at the empirical distribution of previous (or related) instances, or is the (stochastic) output of some ML system.  So in many settings, it is natural to consider a \emph{distribution} as a prediction.  This generalizes the point prediction setting, since the predicted distribution could always be a point distribution.  But can we use these more general predictions?  Motivated by this question, \cite{DILMNV24} studied the Search on Paths (sorted arrays) problem with distributional prediction, and showed that distributional predictions really can be more powerful than point predictions. In this paper we study a significant generalization of their work, namely, the distributional prediction version of Search on Trees.

\begin{definition}\label{def:distributional-search-on-trees}
In the \emph{Distributional Search on Trees} problem, we are given a tree $T=(V,E)$ and there is an unknown target node $t \in V$ drawn from an unknown distribution $p$ over all the vertices. In the prediction setting we are also given a predicted distribution $\widehat{p}$ of \emph{unknown} quality over all the vertices. We can query any node $x \in V$, and that query will return the unique neighbor of $x$ that is on the path from $x$ to $t$. The goal is to find (query) $t$ using as few queries as possible in expectation. 
\end{definition}

It turns out that if $p= \widehat p$ then it is already known that the \emph{weighted} centroid algorithm (the natural generalization of the centroid algorithm where we weight by the distribution) is a $2$-approximation of the instance-optimal search strategy~\citep{BGKK23}.  But in the algorithms with predictions framework we do not want to assume that $p = \widehat p$, and do not even want to assume that we know how ``close'' $p$ and $\widehat p$ are.  So our goal is to design an algorithm for Problem~\ref{def:distributional-search-on-trees} which has expected query complexity which is small when $p$ is close to $\widehat p$ (often called \emph{consistency}), degrades smoothly as $\widehat p$ and $p$ diverge (often called \emph{smoothness}), but never gets too large (often called \emph{robustness}).

\section{Our Results} \label{sec:result}
Before we present the actual algorithm, we summarize what we have already known.

\begin{enumerate}
    \item If the predicted distribution $\widehat{p}$ is accurate, then the weighted centroid algorithm, as discussed in Section~\ref{sec:introduction}, is \emph{asymptotically} optimal. 
    \item If the prediction is just a point, we can run the point prediction algorithm by \cite{dinitz2026treesearchpredictions}.
\end{enumerate}

Of course, both cases on their own seem somewhat lacking for Problem~\ref{def:distributional-search-on-trees}. We are always given a predicted distribution $\widehat{p}$ of \emph{unknown} quality, so running weighted centroid might not actually land us to a node that is relevant. And we certainly should not expect $p$ to be concentrated just on a point.

But what if we combine the two? Inspired by \cite{DILMNV24}, our algorithm (informally) does the following.  In iteration $i$, we perform $2^i$ weighted centroid queries to quickly eliminate regions with large predicted mass, and then run the point prediction algorithm by \cite{dinitz2026treesearchpredictions} for $2^i$ aueries from each of the ``boundary vertices'' (defined more formally in Section~\ref{sec:prelims}) of the remaining feasible subtrees. In this way, our algorithm can exploit $\widehat p$ when it is accurate while still recovering efficiently when the target lies far from where $\widehat p$ predicts.

One important subtlety here, which does not appear in the path case of~\citet{DILMNV24}, is that there can be many boundary vertices.  The natural notion of a ``boundary vertex'' for a subtree would be ``nodes with at least one neighbor outside of the subtree''.  But it is easy to see that after a sequence of weighted-centroid queries, the remaining feasible subtree may have many such ``boundary vertices''.  So naively running our point-prediction algorithm from every one of them would be prohibitively expensive.  This is quite different from the path case, where there can only be two such boundary nodes.    

To get around this obstacle, we show that it is actually possible to ``compress'' these many boundary vertices to only two candidate anchors, without increasing their distance to the target.  After this compression we can simply run our bounded-pathwidth point-prediction algorithm from the two remaining anchors.

To get some high-level idea of how we compress the boundary vertices, we assign unit weight to each boundary vertex and zero weight to every other vertex of the feasible subtree, and query a weighted centroid. If the answer is $\here$, we are done. Otherwise, we update the feasible subtree to be the component indicated by the answer, and update the boundary vertices to be those in the new feasible subtree plus the neighbor returned by the oracle. By the centroid property, this reduces the number of boundary vertices by a constant factor, so after logarithmically many queries we have at most two boundary vertices remaining. These two vertices are our candidate anchors, and we can run our point-prediction algorithm from them. Putting this all together gives the following theorem, which is our main upper bound.

\begin{restatable}{theorem}{DistributionalUpperBound}\label{DistributionalUpperBound}
\label{thm}
There is a deterministic algorithm for the Distributional Search on Trees problem such that, for every tree $T$ of pathwidth at most $k$, every true distribution $p$, and every predicted distribution $\widehat p$, the algorithm finds a target drawn from $p$ using $O\left(H(p)+k\log \eta\right)$ queries in expectation, where $H(p)$ is the entropy of $p$ and $\eta$ is the earth mover's distance between $p$ and $\widehat p$. 
\end{restatable}

\paragraph{Robustness and Running Time.}
Our focus is on query complexity, following~\citet{MitzenmacherVassilvitskii}, ~\citet{DILMNV24}, and ~\citet{dinitz2026treesearchpredictions}. Robustness can be obtained essentially for free by running our algorithm in parallel with the standard centroid-search algorithm and stopping when either one finds the target. This gives query complexity $O(\min\{\log n, H(p) + k \log \eta \})$.
We therefore do not discuss robustness separately in the remainder of the paper. As for computational efficiency, fortunately, \citet{dinitz2026treesearchpredictions} showed that the point-prediction algorithm and the required spine decompositions can be implemented in polynomial time. The additional weighted-centroid and boundary-compression operations used in our distributional algorithm are also polynomial-time, so our overall algorithm remains computationally efficient. Indeed, as we will show in Section~\ref{experiment}, the polynomial-time runtime of our algorithms enables us to run experiments on relatively large graphs.

\section{A Matching Lower Bound}
Given Theorem~\ref{DistributionalUpperBound}, a natural next question is
whether its dependence on entropy and prediction error is necessary. We
show that both terms are independently necessary. For the prediction-error
term, we reduce the point-prediction problem to the distributional setting
and invoke the lower bound of~\citet{dinitz2026treesearchpredictions}, which
shows that, for trees of pathwidth $k$, there are prediction--target pairs
$(s,t)$ for which every deterministic or zero-error randomized algorithm
requires $    \Omega\!\left(k\log \dist(s,t)\right)$
queries in expectation.

The lower-bound construction of~\citet{dinitz2026treesearchpredictions}
fixes a tree and a predicted vertex $s$, and considers a hard distribution
over possible targets in this tree. We associate each possible target $t$
with a separate distributional-search instance, whose true distribution
$p=\delta_t$ places unit mass at $t$ and whose predicted distribution
$\widehat p=\delta_s$ places unit mass at $s$. The original hard
distribution over targets therefore induces a distribution over these
instances. Within each individual instance, however, the target is
deterministic, so $H(p)=0$, and
$
    \eta=\EMD(\delta_t,\delta_s)=\dist(s,t).
$

Since the true distribution is unknown to the algorithm, and every
instance presents the same tree and predicted distribution, identifying
the target requires exactly the same oracle queries as in the
point-prediction setting. The original lower bound therefore carries
over: every deterministic or zero-error randomized algorithm has some
distributional-search instance on which it requires
$\Omega(k\log\eta)$ queries in expectation, even though the true
distribution in that instance has zero entropy.

The entropy term is also necessary, even when the prediction is exact.
Taking $\widehat p=p$ on a path recovers the classical setting of binary
search with a known target distribution, where the lower bound
of~\citet{Mehlhorn1975} gives expected query complexity $\Omega(H(p))$.
Thus, both the entropy term and the prediction-error term are
independently necessary, yielding the following theorem.

\begin{restatable}{theorem}{DistributionallowerBound}
\label{thm:DistributionalLowerBound}
For every $k\ge 1$, there exists an infinite family of instances of
\textsc{Distributional Search on Trees} on trees of pathwidth $k$ and
maximum degree at most $3$ such that every deterministic or zero-error
randomized algorithm has some instance in the family on which it requires
$
    \Omega\left(
        H(p)
        +
        k\log\left( \eta \right)
    \right)
$
queries in expectation, where
$
    \eta=\EMD(p,\widehat p)
$
is the earth mover's distance between the true distribution $p$ and the
predicted distribution $\widehat p$. Moreover, the instances witnessing
the prediction-error term can be chosen so that $H(p)=0$.
\end{restatable}

\section{Related Work} \label{sec:related}

\paragraph{Learning-Augmented Algorithms.}
Learning-augmented algorithms have received substantial attention as machine-learned predictions have become increasingly available in algorithmic applications. Following the influential work of \citet{lykouris2021competitive}, this framework has been explored across a wide range of settings, including online problems such as ski rental~\citep{Purohit}, scheduling~\citep{LattanziLMV, doi:10.1137/1.9781611978971.47}, knapsack~\citep{ImKQP21, lechowicz2024time, pmlr-v267-daneshvaramoli25a}, set cover~\citep{BamasMS20}, It has also been applied to accelerate classical combinatorial algorithms~\citep{DinitzILMV21}, design dynamic algorithms~\citep{BrandFNP24}, study mechanism-design problems~\citep{AgrawalBGOT24} and computational geometry \citep{10.1007/978-3-031-81396-2_3, cabello_et_al:LIPIcs.ITCS.2026.31}, among many other direction.

\paragraph{Algorithms with Distributional Predictions.}
Among the many forms of predictions considered in learning-augmented algorithms, distributional predictions have also received increasing attention. This model has been studied in several classical algorithmic settings, including binary search \citep{DILMNV24}, ski rental \citep{cui2026skirentaldistributionalpredictions, kang2025learningaugmentedskirentaldiscrete, kim2026robustconsistentskirental}, and contract scheduling \citep{angelopoulos2024contractschedulingdistributionalmultiple}.

\paragraph{Search in trees and graphs.}
There is also a substantial literature on search beyond paths in the absence of predictions. \citet{OP06} studied search on trees and showed that centroid search uses $O(\log n)$ queries. A related problem is ``Search Trees on Trees'' (STT), where a distribution over the vertices is known exactly and the goal is to optimize the resulting search strategy. In this setting, \citet{BGKK23} showed that weighted-centroid search gives a $2$-approximation to the optimum. Related search models have also been studied on more general graphs~\citep{EKS16}.
More recently, \citet{dinitz2026treesearchpredictions} studied Search on Trees with point predictions, where the algorithm is given a predicted target vertex. They showed that the $O(\log \dist(s,t))$ guarantee available on paths cannot hold on arbitrary trees, but that for trees of pathwidth $k$, a recursive spine-based algorithm achieves $O(k\log \dist(s,t))$ query complexity, with a matching asymptotic lower bound. The present work extends this line of study to distributional predictions, where the algorithm is given an arbitrary predicted distribution rather than a single predicted vertex.

\section{Preliminaries and Notation}
\label{sec:prelims}

Given a tree $T$, we denote its vertex set by $V(T)$ and its edge set by $E(T)$. For vertices $u,v\in V(T)$, let $\dist_T(u,v)$ denote the length of the unique simple path between $u$ and $v$ in $T$. If $P$ is a path in $T$ and $u,v\in V(P)$, then $\dist_P(u,v)$ denotes the number of edges on the subpath of $P$ between $u$ and $v$.

We first formalize the oracle model used throughout the paper.

\begin{definition}[Direction oracle]
\label{def:direction-oracle}
Let $T=(V(T),E(T))$ be a tree, and let $t\in V(T)$ be a hidden target vertex. The direction oracle for $t$ is the map
$
    \mathrm{dir}_t:V(T)\to V(T)\cup\{\textsf{here}\}
$
defined as follows. If $v=t$, then $\mathrm{dir}_t(v)=\textsf{here}$.  Otherwise, $\mathrm{dir}_t(v)$ is the unique neighbor $u$ of $v$ that lies on the simple path from $v$ to $t$ in $T$. An oracle query at $v$ returns $\mathrm{dir}_t(v)$.
\end{definition}

In the \emph{Distributional Tree Search} problem, the algorithm is given a predicted distribution $\widehat p\in\Delta(V(T))$, while the target is
drawn from an unknown true distribution $p\in\Delta(V(T))$.
 The algorithm may adaptively query vertices of $T$ through the direction oracle $\mathrm{dir}_t$. The algorithm succeeds when it outputs $t$; equivalently, it may stop once it queries a vertex $v$ with $\mathrm{dir}_t(v)=\textsf{here}$. The query complexity of an algorithm is the number of oracle queries it makes. When the target is clear from context, we write $\mathrm{dir}$ instead of $\mathrm{dir}_t$.

Now we recall the standard notion of earth mover's distance (optimal transport). 

\begin{definition}[Earth Mover's Distance]
    A coupling of $p$ and $\widehat p$ is a nonnegative matrix
$\pi\in\mathbb R_{\ge 0}^{V\times V}$ satisfying $\sum_{y\in V}\pi(x,y)=p_x$ and $\sum_{x\in V}\pi(x,y)=\widehat p_y.$ Let $\Pi(p,\widehat p)$ be the set of all such couplings. The earth mover's distance in the tree metric is $\Wone(p,\widehat p)
    :=
    \min_{\pi\in\Pi(p,\widehat p)}
    \sum_{x,y\in V}\pi(x,y)\dist_T(x,y)$.
\end{definition}

Now we consider the definition of boundary vertices of a subtree.

\begin{definition}
The boundary set of a subtree $F$ of $T$ is the set 
    $
\boundary F
:=
\{
x\in V(F):
x\text{ has a neighbor in }V(T)\setminus V(F)
\}.
$
\end{definition}

Finally, we will use the standard notion of distance from a node to a set of nodes as follow:

\begin{definition}
    For any $A\subseteq V(T)$ and node $t \in V(T)$, the distance from $t$ to $A$ is 
$
\dist_T(t,A)
:=
\min_{a\in A}\dist_T(t,a).
$
\end{definition}

\section{Main Algorithm and Analysis}\label{sec:dist-up-bound}
In this section we describe our \emph{distributional} tree search algorithm and prove it's query complexity.

The algorithm alternates two types of work. It first makes weighted-centroid queries using the entire predicted distribution $\widehat p$ on the current feasible subtree. These queries rapidly reduce the absolute predicted mass of the feasible subtree and are responsible for the entropy term. It then searches geometrically from the boundary of the feasible subtree. Since a subtree can have many boundary vertices, our algorithm therefore first need to know how to efficiently compress the boundary into at most $2$ nodes.

\subsection{The Boundary Compression Sub-routine}

Recall that a weighted centroid of a tree with nonnegative vertex weights is a vertex such that every component remaining after its deletion has at most half of the total weight. So to reduce the number of the boundary points, we could assign equal weights to all the boundary points and zero weight to all other nodes. Running weighted centroid on these nodes gives the desired result. 

\begin{lemma}[Boundary compression]\label{lem:boundary-compression}
Let $F$ be a connected subtree containing $t$, and let $A=\boundary F$.
There is an adaptive procedure that, using $O(\log(|A|))$ queries, either
finds $t$ or produces a connected subtree $F'\subseteq F$ containing $t$ and
a set $A'\subseteq V(F')$ such that
$
    1\le |A'|\le 2
    \qquad\text{and}\qquad
    \dist_T(t,A')\le \dist_T(t,A).
$
\end{lemma}

\begin{proof}
While $|A|>2$, assign unit weight to every vertex of $A$ and zero weight to
every other vertex of $F$, and query a weighted centroid $v$ of $F$. If
$v=t$, stop. Otherwise let $y$ be the returned neighbor, and let $C$ be the
component of $F-v$ containing $y$ and hence containing $t$. Update
$
    F\leftarrow C,
    A\leftarrow (A\cap C)\cup\{y\}.
$

If $m$ is the old size of $A$, the centroid property gives
$
    |A_{\mathrm{new}}|\le \frac{m}{2}+1,
$
and therefore
$
    |A_{\mathrm{new}}|-2
    \le
    \frac{|A_{\mathrm{old}}|-2}{2}.
$
Hence $O(\log(|A|))$ queries suffice to reduce $A$ to at most two vertices.

It remains to prove the distance invariant. Let $a\in A$ be closest to $t$.
If $a\in C$, then it remains in the new anchor set. Otherwise, the unique
$a$--$t$ path passes through $v$ and then through $y$, so $y$ is closer to
$t$ than $a$. Thus the distance from $t$ to the anchor set never increases.
\end{proof}

\subsection{The Distributional Algorithm}

Now we give the pseudocode for our distributional tree search algorithm with predictions. For stage $i = 0, 1, 2, ...$ we do $2^i$ iterations of weighted centroid on the remaining feasible subtree. Then we do the  boundary reduction and run $2^i$ steps of the point prediction algorithm. Increasing the step size geometrically, the algorithm stops until it finds the target. 

See Algorithm~\ref{alg:main} for full detail.

\begin{algorithm}[t]
\caption{Distributional Tree Search}
\label{alg:main}
\begin{algorithmic}[1]
\State $F\leftarrow T$
\For{stage $i=0,1,2,\ldots$}
    \For{$2^i$ iterations}
        \If{$F$ is a singleton}
            \State \Return its unique vertex
        \EndIf

        \State Choose a $\widehat p$-weighted centroid (or centroid if no mass is left) $v$ of $F$ and query $v$

        \If{the answer is $\here$}
            \State \Return $v$
        \EndIf

        \State Replace $F$ by the component of $F-v$ indicated by the answer
    \EndFor

    \State $A\leftarrow \boundary F$
    \State Run the compression procedure of
    \Cref{lem:boundary-compression}, updating $F$ and $A$, until $|A|\le 2$

    \State Let $U\leftarrow F$ and independently initialize the point-prediction algorithm on $(U,a)$ for every $a\in A$

    \State Execute these at most two procedures in round-robin order for a
    total of $2^i$ queries

    \If{one of the procedures finds the target}
        \State \Return the target
    \EndIf

    \State Replace $F$ by the subtree of $U$ consistent with all answers
    obtained during the recovery phase
\EndFor
\end{algorithmic}
\end{algorithm}

\subsection{Analysis}
Now we analyze Algorithm~\ref{alg:main}. We first note that the geometric increase in step size ``absorbs'' all the cost of the algorithm.

\begin{lemma}[Geometric stage cost]\label{lem:stage-cost}
Let $j \in \{0, 1, 2, ...\}$ be the first stage in which \Cref{alg:main} finds target $t$.
Then the algorithm finds $t$ using $O(2^j)$ queries.
\end{lemma}

\begin{proof}
A query made at a vertex of the current feasible subtree can add at most one
new vertex to its inner boundary, while a query outside the feasible subtree
cannot increase the boundary. Hence the boundary size is at most the number
of preceding queries.

Suppose inductively that $O(2^i)$ queries have been made before stage $i$.
The bisection part and the truncated point-recovery part each use $2^i$
queries. The boundary has size $O(2^i)$, so the boundary reduction step of 
\Cref{lem:boundary-compression} uses $O(i)$ queries. Therefore stage $i$ costs $O(2^i)$, and summing the geometric series through stage $j$ shows that the algorithm finds $t$ using
$O(2^j)$ queries.
\end{proof}

We next connect the running time of the algorithm to either the entropy term or the prediction error. Intuitively, the algorithm can only spend many queries on a target \(t\) if the predicted distribution is substantially inaccurate around \(t\). If the algorithm reaches a sufficiently large stage, then the repeated weighted-centroid queries have already reduced the predicted mass of the feasible subtree containing \(t\) to below \(p_t/2\). Moreover, since the recovery phase of the preceding stage still failed, the point-prediction guarantee implies that \(t\) must lie sufficiently far from the boundary of this subtree. Thus either \(t\) is found after only \(O(\log(1/p_t))\) queries, or we obtain a subtree containing \(t\) that has both too little predicted mass and a sufficiently large distance from \(t\) to its boundary. These properties will allow us to charge the search cost to the earth mover's distance.

\begin{restatable}{lemma}{DichotomyLemma} \label{lem:dichotomy}
For every target $t$ with $p_t>0$, at least one of the following holds:
\begin{enumerate}[(i)]
    \item the algorithm finds $t$ using
    $
        O\left(\log\frac{4}{p_t}\right)
    $
    queries;
    \item there is a connected subtree $S_t\subseteq T$ containing $t$ such
    that
    $
        \widehat p(S_t)\le \frac{p_t}{2}
    $
    and, writing
    $
        \delta_t:=\dist_T(t,\boundary S_t),
    $
    the algorithm finds $t$ using
    $
        O\!\left(k\log \delta_t\right)
    $
    queries.
\end{enumerate}
\end{restatable}

\begin{proof}

Let $j$ be the stage where $t$ is found. If $2^j\le \log\frac{4}{p_t}$, then part~(i) follows immediately from \Cref{lem:stage-cost}.

Assume instead that $2^j>\log\frac{4}{p_t}$, which implies $2^{-2^j} < \frac{p_t}{4}$. Let $S_t$ be the feasible subtree immediately after
the weighted centroid part of stage $j-1$. Through this point the algorithm has made
$
    1+2+\cdots+2^{j-1}=2^j-1
$
weighted-centroid queries. Each such query halves the absolute predicted
mass of the feasible subtree, so
$
    \widehat p(S_t)
    \le
    2^{-(2^j-1)}
    =
    2^{1-2^j}
    <
    \frac{p_t}{2}.
$

Let \(A'\) be the set of at most two anchors produced by boundary compression in stage \(j-1\). By \Cref{lem:boundary-compression}, there exists some \(a\in A'\) such that
$
    \dist_T(t,a)
    \le
    \dist_T(t,\boundary S_t)
    =
    \delta_t.
$

During the recovery phase of stage \(j-1\), the algorithm executes the point-prediction procedures initialized from the vertices of \(A'\) in round-robin order for a total of \(2^{j-1}\) queries. Since \(|A'|\le 2\), the procedure initialized from \(a\) receives at least
$
    2^{j-2}-1
$
queries. Because \(j\) is the first stage in which the target is found, this procedure does not terminate within this budget.

Let \(Q(a,t)\) denote the number of queries that the point prediction algorithm, initialized with prediction \(a\), requires to find \(t\). The preceding observation implies
$
    Q(a,t) > 2^{j-2}-1.
$
On the other hand, simply invoking the query complexity of the point prediction algorithm gives
$
    Q(a,t)
    =
    O\!\left(
        k\log \dist_T(a,t)
    \right)
    \le
    O\!\left(
        k\log\delta_t
    \right).
$
Consequently,
$
    2^j
    =
    O\!\left(
        k\log \delta_t
    \right).
$

Finally, \Cref{lem:stage-cost} implies that the total number of queries made before finding \(t\) is \(O(2^j)\), and hence
$
    O\!\left(
        k\log\frac{\delta_t}{k}
    \right).
$
This proves part~(ii).
\end{proof}

Now we show the second term of \Cref{lem:dichotomy} is actually \emph{on average} bounded by a earth mover's distance term between $p$ and $\widehat{p}$.

\begin{lemma}[Optimal Transport Cost]\label{lem:transport-charge}
Let $L$ be the set of targets satisfying part~(ii) of
\Cref{lem:dichotomy}, and let $\eta=\Wone(p,\widehat p)$. Then
$
    \sum_{t\in L}p_t k\log \delta_t
    \le
    2k\log\eta.
$
\end{lemma}

\begin{proof}
Let $\pi^*$ be an optimal coupling for $\Wone(p,\widehat p)$. For a target
$t\in L$, the amount of row-$t$ mass that can be transported to vertices
inside $S_t$ is at most the total destination mass of $S_t$:
$
    \sum_{y\in S_t}\pi^*(t,y)
    \le
    \sum_{x\in V}\sum_{y\in S_t}\pi^*(x,y)
    =
    \widehat p(S_t)
    \le
    \frac{p_t}{2}.
$
Thus at least $p_t/2$ units of mass from $t$ are transported outside $S_t$.
Every path from $t$ to a vertex outside $S_t$ crosses $\boundary S_t$, and
hence every such destination has distance at least $\delta_t$ from $t$.
Since the logarithm is increasing,
$
    \sum_{y\in V}
    \pi^*(t,y)k\log\left(\dist_T(t,y)\right)
    \ge
    \frac{p_t}{2}k\log\delta_t.
$
Summing over $t\in L$ and using Jensen's inequality gives
\begin{align*}
    \sum_{t\in L}p_t k\log\delta_t
    &\le
    2\sum_{x,y\in V}
    \pi^*(x,y)k\log\left(\dist_T(x,y)\right)\\
    &\le
    2k\log\left(
       \sum_{x,y\in V}\pi^*(x,y)\dist_T(x,y)
    \right)\\
    &=
    2k\log\eta.
\end{align*}
The sets $S_t$ may overlap; this causes no double counting because each
target is charged only through its own source row in the coupling.
\end{proof}

We are now finally ready to prove our main theorem. 

\DistributionalUpperBound*

\begin{proof}
Partition the targets into the two cases of \Cref{lem:dichotomy}. For targets
of type~(i), their contribution to the expected number of queries is
$
    O\left(
        \sum_{t\in V}p_t\log\frac{4}{p_t}
    \right)
    =
    O(H(p)).
$

For targets of type~(ii),
\Cref{lem:dichotomy,lem:transport-charge} show that their contribution to the
expected number of queries is
$
    O\left(
        \sum_{t\text{ of type (ii)}}p_t k\log\delta_t
    \right)
    =
    O\!\left(k\log\eta\right).
$
Adding the two bounds shows that the algorithm finds a target drawn from $p$
using
$
    O\!\left(
        H(p)+k\log\eta
    \right)
$
queries in expectation.
\end{proof}

\section{Lower Bound}

In this section, we show that both terms in the upper bound from
Theorem~\ref{DistributionalUpperBound} are independently necessary.
The entropy lower bound follows from classical binary search with a
known target distribution. For the prediction-error term, we reduce
the point-prediction problem to the distributional setting and apply
the lower bound of~\citet{dinitz2026treesearchpredictions}. This reduction
yields hard instances whose true distributions have zero entropy.
\subsection{An Entropy Lower Bound}

We first observe that the $\Omega(H(p))$ term follows immediately from the classical lower bound for binary search. Indeed, a path is a special case of a tree, and our query model on a path is precisely comparison-based binary search. Mehlhorn~\citep{Mehlhorn1975} showed that, even when the access distribution $p$ is known, every binary search strategy has expected query complexity at least $\Omega(H(p))$. Taking $\widehat p=p$, this immediately gives an $\Omega(H(p))$ lower bound for Distributional Search on Trees, even when the predicted distribution is perfectly accurate. It therefore remains to show that the dependence on the prediction error is also necessary, even when $H(p)=0$.

\subsection{A Prediction-Error Lower Bound}

We now show that the dependence on prediction error and pathwidth is
necessary even when $H(p)=0$.

\DistributionallowerBound*

\begin{proof}
Fix $k\ge 1$ and an integer $\ell\ge\max\{4,k^2\}$. We use the
construction of~\citet{dinitz2026treesearchpredictions}, which provides
a tree $T_{k,\ell}$ of pathwidth $k$ and maximum degree at most $3$,
a fixed predicted vertex $s$, and a set $B_{k,\ell}$ of admissible
targets. Every $t\in B_{k,\ell}$ satisfies
$k\ell/2\le\dist(s,t)\le k\ell$. Let $\mu$ be the uniform distribution
over $B_{k,\ell}$. Their decision-tree lower bound shows that every
deterministic search algorithm $C$ that correctly identifies every
admissible target satisfies
$\mathbb{E}_{t\sim\mu}[Q_C(t)]\ge c k\log\ell$,
where $Q_C(t)$ is its number of oracle queries on target $t$, and
$c>0$ is an absolute constant.

For each $t\in B_{k,\ell}$, define a separate distributional-search
instance $\mathcal{I}_t$ on $T_{k,\ell}$. Its true distribution
$p^{(t)}=\delta_t$ places unit mass at $t$, and its predicted
distribution $\widehat p=\delta_s$ places unit mass at $s$.
Let $\nu$ be the distribution over these instances obtained by
drawing $t\sim\mu$ and selecting $\mathcal{I}_t$.
Thus, $\nu$ determines which instance is selected, while
$p^{(t)}$ determines the target within the selected instance.
Since the latter distribution is concentrated at a single vertex,
every instance satisfies $H(p^{(t)})=0$ and
$\eta_t:=\EMD(p^{(t)},\widehat p)=\dist(s,t)$.

Let $A$ be any deterministic or zero-error randomized algorithm for
Distributional Search on Trees. We obtain a point-prediction
algorithm $A'$ by giving $A$ the tree $T_{k,\ell}$ and predicted
distribution $\delta_s$, and forwarding each query to the
point-search oracle. This simulation does not require knowledge
of $t$ or $p^{(t)}$, because the true distribution is not an input
to $A$. For every target $t$ and every fixed choice $\rho$ of
the algorithm's random bits, the simulation produces exactly the
same queries and oracle responses as running $A$ on
$\mathcal{I}_t$. In particular,
$Q_{A'}(t;\rho)=Q_A(\mathcal{I}_t;\rho)$.

For a deterministic algorithm, the cited decision-tree lower bound
therefore gives
$\mathbb{E}_{t\sim\mu}[Q_A(\mathcal{I}_t)]\ge c k\log\ell$.
For a zero-error randomized algorithm, the finite number of admissible
targets ensures that, for almost every fixed choice of $\rho$,
the simulated algorithm is correct on every target in $B_{k,\ell}$.
Applying the same deterministic lower bound and then averaging over
$\rho$ gives
$\mathbb{E}_{t\sim\mu}\mathbb{E}_{\rho}
[Q_A(\mathcal{I}_t;\rho)]\ge c k\log\ell$.
Equivalently, this is a lower bound on the expected query cost when
the instance is drawn from $\nu$.

By averaging, there is some $t\in B_{k,\ell}$ for which
$\mathbb{E}_{\rho}[Q_A(\mathcal{I}_t;\rho)]\ge c k\log\ell$.
Since $\ell\ge k^2$, we have
$\eta_t=\dist(s,t)\le k\ell\le\ell^{3/2}$, and hence
$k\log\ell\ge\frac{2}{3}k\log\eta_t$.
Consequently, $A$ requires $\Omega(k\log\eta_t)$ queries in
expectation on $\mathcal{I}_t$, even though $H(p^{(t)})=0$.
The deterministic case follows with the expectation over $\rho$
omitted.

Allowing $\ell$ to range over all integers at least
$\max\{4,k^2\}$ gives the required infinite family.
Together with the entropy lower bound for exact predictions, this
establishes that both terms in the upper bound are independently
necessary.
\end{proof}

\section{Experimental Setup}
\label{experiment}

We evaluate how errors in a predicted target distribution affect the
query complexity of tree search. Our experiments use four trees:
two spanning trees derived from road networks and two artificial trees
with prescribed pathwidth. On each tree, we gradually replace the
support of the predicted distribution while keeping the true
distribution fixed. This provides a controlled comparison across five
levels of support mismatch and allows us to examine both the overhead
of our algorithm under accurate predictions and its behavior as the
prediction becomes less accurate.

\paragraph{Datasets and tree construction.}
The two road networks are the Luxembourg road network and the US road
network covering the 48 contiguous states, obtained from \cite{nr}.
We interpret each input as an undirected, unweighted simple graph,
removing self-loops and duplicate edges. We extract its largest
connected component and construct a BFS spanning tree, starting from
the vertex with the smallest original identifier and visiting neighbors
in increasing identifier order. This makes the tree construction
deterministic. All resulting tree edges have unit length.

We also construct two artificial trees with exact pathwidths 2 and 3 using the hard instance construction by \cite{dinitz2026treesearchpredictions}.

Table~\ref{tab:tree-statistics} summarizes the four experimental trees.
The road trees provide examples derived from network data, while the
artificial trees allow us to evaluate the algorithms on larger trees
with explicitly controlled pathwidth.

\begin{table}[t]
    \centering
    \small
    \setlength{\tabcolsep}{4pt}
    \begin{tabular}{lrrrrr}
        \toprule
        Tree & Vertices & Edges & Diameter & Pathwidth & Max.\ degree \\
        \midrule
        Luxembourg road
            & 114,599 & 114,598 & 1,985 & 6 & 5 \\
        US road
            & 126,146 & 126,145 & 936 & 7 & 6 \\
        Artificial, pathwidth 2
            & 1,048,576 & 1,048,575 & 3,069 & 2 & 3 \\
        Artificial, pathwidth 3
            & 1,050,504 & 1,050,503 & 503 & 3 & 3 \\
        \bottomrule
    \end{tabular}
    \caption{Properties of the four experimental trees.
    All edges have unit length, diameters are measured in edges,
    and all reported pathwidth values are exact.}
    \label{tab:tree-statistics}
\end{table}

\paragraph{True and predicted distributions.}
For a tree with $n$ vertices, we sample $\lfloor n/2\rfloor$
vertices uniformly without replacement and define the true
distribution to be uniform on this support. Within each trial,
we keep the true distribution fixed and construct predicted
distributions by progressively replacing support vertices with
randomly sampled vertices outside the true support. Each predicted
distribution is uniform on its support, whose size remains unchanged.
This produces predictions ranging from an exact match to distributions
with disjoint supports.

This experiment models a distribution shift between predicted and
actual target locations. A prediction may correctly identify some
likely locations while assigning probability to others that do not
match the true distribution. Varying the overlap between the supports
provides a controlled way to adjust this mismatch and examine how
search performance changes as the prediction becomes less informative.
Keeping the support size and uniform weights fixed ensures that
the comparison reflects changes in the predicted locations rather
than changes in how concentrated the distributions are.

\paragraph{Algorithms and query accounting.}
We compare our algorithm with two versions of weighted centroid
search, using either the true distribution or the predicted
distribution.

The true-distribution baseline uses $p$ as the vertex weights.
At each step, it queries a weighted centroid of the current feasible
subtree, so that every component remaining after removing the queried
vertex contains at most half of the remaining probability mass.
Unless the queried vertex is the target, the oracle response
identifies the component in which the search continues.
Because this baseline knows the true target distribution, it serves
as an informed benchmark for evaluating the cost of using an
imperfect prediction.

The predicted-distribution baseline initially follows the same
procedure using the predicted weights $\widehat p$. While the feasible
subtree contains positive predicted mass, it queries a weighted
centroid and restricts the search according to the oracle response.
However, an inaccurate prediction may assign zero probability to
the entire remaining feasible subtree. If this happens before the
target is resolved, the baseline stops using the prediction and
restarts ordinary centroid search from the remaining feasible subtree tree.
In this second phase, vertices receive equal weight, and each
centroid query reduces the number of remaining candidate vertices.

For all methods, each direction-oracle call costs one query, and
a singleton feasible subtree is resolved without an additional
confirmation query. Appendix~\ref{experimental results} records additional experiments details and reports the
experimental results.

\newpage
\subsection*{AI use statement}
No artificial intelligence tools were used in the preparation of this work. All ideas, proofs, and written content were developed and written independently by the authors.

\bibliography{sources}

@article{BENTLEY197682,
title = {An almost optimal algorithm for unbounded searching},
journal = {Information Processing Letters},
volume = {5},
number = {3},
pages = {82-87},
year = {1976},
issn = {0020-0190},
doi = {https://doi.org/10.1016/0020-0190(76)90071-5},
url = {https://www.sciencedirect.com/science/article/pii/0020019076900715},
author = {Jon Louis Bentley and Andrew Chi-Chih Yao}
}

@inproceedings{nr,
      title = {The Network Data Repository with Interactive Graph Analytics and Visualization},
      author={Ryan A. Rossi and Nesreen K. Ahmed},
      booktitle = {AAAI},
      url={https://networkrepository.com},
      year={2015}
  }

@inbook{MitzenmacherVassilvitskii, 
place={Cambridge}, 
title={Algorithms with Predictions}, 
DOI={10.1017/9781108637435.037}, 
booktitle={Beyond the Worst-Case Analysis of Algorithms}, 
publisher={Cambridge University Press}, 
author={Mitzenmacher, Michael and Vassilvitskii, Sergei}, 
editor={Roughgarden, Tim}, 
year={2021}, 
pages={646–662}}

@inproceedings{DILMNV24,
  author       = {Michael Dinitz and
                  Sungjin Im and
                  Thomas Lavastida and
                  Benjamin Moseley and
                  Aidin Niaparast and
                  Sergei Vassilvitskii},
  editor       = {Amir Globersons and
                  Lester Mackey and
                  Danielle Belgrave and
                  Angela Fan and
                  Ulrich Paquet and
                  Jakub M. Tomczak and
                  Cheng Zhang},
  title        = {Binary Search with Distributional Predictions},
  booktitle    = {Advances in Neural Information Processing Systems 38: Annual Conference
                  on Neural Information Processing Systems 2024, NeurIPS 2024, Vancouver,
                  BC, Canada, December 10 - 15, 2024},
  year         = {2024},
  url          = {http://papers.nips.cc/paper\_files/paper/2024/hash/a4b293979b8b521e9222d30c40246911-Abstract-Conference.html},
  bibsource    = {dblp computer science bibliography, https://dblp.org}
}

@article{lykouris2021competitive,
  title={Competitive caching with machine learned advice},
  author={Lykouris, Thodoris and Vassilvitskii, Sergei},
  journal={Journal of the ACM (JACM)},
  volume={68},
  number={4},
  pages={1--25},
  year={2021},
  publisher={ACM New York, NY}
}

@inproceedings{BrandFNP24,
  author       = {Jan van den Brand and
  Sebastian Forster and
  Yasamin Nazari and
  Adam Polak},
  title        = {On Dynamic Graph Algorithms with Predictions},
  booktitle    = {{SODA}},
  pages        = {3534--3557},
  publisher    = {{SIAM}},
  year         = {2024}
}

@article{AgrawalBGOT24,
  author       = {Priyank Agrawal and
                  Eric Balkanski and
                  Vasilis Gkatzelis and
                  Tingting Ou and
                  Xizhi Tan},
  title        = {Learning-Augmented Mechanism Design: Leveraging Predictions for Facility
                  Location},
  journal      = {Math. Oper. Res.},
  volume       = {49},
  number       = {4},
  pages        = {2626--2651},
  year         = {2024},
  url          = {https://doi.org/10.1287/moor.2022.0225},
  doi          = {10.1287/MOOR.2022.0225},
  bibsource    = {dblp computer science bibliography, https://dblp.org}
}

@inproceedings{Purohit,
  title={Improving online algorithms via ml predictions},
  author={Purohit, Manish and Svitkina, Zoya and Kumar, Ravi},
  booktitle={Advances in Neural Information Processing Systems},
  pages={9661--9670},
  year={2018}
}

@inproceedings{LattanziLMV,
  author    = {Silvio Lattanzi and
               Thomas Lavastida and
               Benjamin Moseley and
               Sergei Vassilvitskii},
  editor    = {Shuchi Chawla},
  title     = {Online Scheduling via Learned Weights},
  booktitle = {Proceedings of the 2020 {ACM-SIAM} Symposium on Discrete Algorithms,
               {SODA} 2020, Salt Lake City, UT, USA, January 5-8, 2020},
  pages     = {1859--1877},
  publisher = {{SIAM}},
  year      = {2020},
  url       = {https://doi.org/10.1137/1.9781611975994.114},
  doi       = {10.1137/1.9781611975994.114},
  bibsource = {dblp computer science bibliography, https://dblp.org}
}

@inproceedings{ImKQP21,
  author       = {Sungjin Im and
                  Ravi Kumar and
                  Mahshid Montazer Qaem and
                  Manish Purohit},
  editor       = {Marc'Aurelio Ranzato and
                  Alina Beygelzimer and
                  Yann N. Dauphin and
                  Percy Liang and
                  Jennifer Wortman Vaughan},
  title        = {Online Knapsack with Frequency Predictions},
  booktitle    = {Advances in Neural Information Processing Systems 34: Annual Conference
                  on Neural Information Processing Systems 2021, NeurIPS 2021, December
                  6-14, 2021, virtual},
  pages        = {2733--2743},
  year         = {2021},
  url          = {https://proceedings.neurips.cc/paper/2021/hash/161c5c5ad51fcc884157890511b3c8b0-Abstract.html},
  bibsource    = {dblp computer science bibliography, https://dblp.org}
}

@inproceedings{BamasMS20,
  author    = {{\'{E}}tienne Bamas and
               Andreas Maggiori and
               Ola Svensson},
  editor    = {Hugo Larochelle and
               Marc'Aurelio Ranzato and
               Raia Hadsell and
               Maria{-}Florina Balcan and
               Hsuan{-}Tien Lin},
  title     = {The Primal-Dual method for Learning Augmented Algorithms},
  booktitle = {Advances in Neural Information Processing Systems 33: Annual Conference
               on Neural Information Processing Systems 2020, NeurIPS 2020, December
               6-12, 2020, virtual},
  year      = {2020},
  url       = {https://proceedings.neurips.cc/paper/2020/hash/e834cb114d33f729dbc9c7fb0c6bb607-Abstract.html},
  bibsource = {dblp computer science bibliography, https://dblp.org}
}

@inproceedings{DinitzILMV21,
  author    = {Michael Dinitz and
               Sungjin Im and
               Thomas Lavastida and
               Benjamin Moseley and
               Sergei Vassilvitskii},
  editor    = {Marc'Aurelio Ranzato and
               Alina Beygelzimer and
               Yann N. Dauphin and
               Percy Liang and
               Jennifer Wortman Vaughan},
  title     = {Faster Matchings via Learned Duals},
  booktitle = {Advances in Neural Information Processing Systems 34: Annual Conference
               on Neural Information Processing Systems 2021, NeurIPS 2021, December
               6-14, 2021, virtual},
  pages     = {10393--10406},
  year      = {2021},
  url       = {https://proceedings.neurips.cc/paper/2021/hash/5616060fb8ae85d93f334e7267307664-Abstract.html},
  bibsource = {dblp computer science bibliography, https://dblp.org}
}

@InProceedings{BGKK23,
  author =	{Berendsohn, Benjamin Aram and Golinsky, Ishay and Kaplan, Haim and Kozma, L\'{a}szl\'{o}},
  title =	{{Fast Approximation of Search Trees on Trees with Centroid Trees}},
  booktitle =	{50th International Colloquium on Automata, Languages, and Programming (ICALP 2023)},
  pages =	{19:1--19:20},
  series =	{Leibniz International Proceedings in Informatics (LIPIcs)},
  ISBN =	{978-3-95977-278-5},
  ISSN =	{1868-8969},
  year =	{2023},
  volume =	{261},
  editor =	{Etessami, Kousha and Feige, Uriel and Puppis, Gabriele},
  publisher =	{Schloss Dagstuhl -- Leibniz-Zentrum f{\"u}r Informatik},
  address =	{Dagstuhl, Germany},
  URL =		{https://drops.dagstuhl.de/entities/document/10.4230/LIPIcs.ICALP.2023.19},
  URN =		{urn:nbn:de:0030-drops-180711},
  doi =		{10.4230/LIPIcs.ICALP.2023.19}
}

@INPROCEEDINGS{OP06,
  author={Onak, Krzysztof and Parys, Pawel},
  booktitle={2006 47th Annual IEEE Symposium on Foundations of Computer Science (FOCS'06)}, 
  title={Generalization of Binary Search: Searching in Trees and Forest-Like Partial Orders}, 
  year={2006},
  volume={},
  number={},
  pages={379-388},
  doi={10.1109/FOCS.2006.32}}

@inproceedings{EKS16,
author = {Emamjomeh-Zadeh, Ehsan and Kempe, David and Singhal, Vikrant},
title = {Deterministic and probabilistic binary search in graphs},
year = {2016},
isbn = {9781450341325},
publisher = {Association for Computing Machinery},
address = {New York, NY, USA},
url = {https://doi.org/10.1145/2897518.2897656},
doi = {10.1145/2897518.2897656},
booktitle = {Proceedings of the Forty-Eighth Annual ACM Symposium on Theory of Computing},
pages = {519–532},
numpages = {14},
location = {Cambridge, MA, USA},
series = {STOC '16}
}

@article{Mehlhorn1975,
  author  = {Kurt Mehlhorn},
  title   = {Nearly Optimal Binary Search Trees},
  journal = {Acta Informatica},
  volume  = {5},
  pages   = {287--295},
  year    = {1975},
  doi     = {10.1007/BF00264563}
}

@misc{dinitz2026treesearchpredictions,
      title={Tree Search With Predictions}, 
      author={Michael Dinitz and Bob Dong},
      year={2026},
      eprint={2605.27490},
      archivePrefix={arXiv},
      primaryClass={cs.DS},
      url={https://arxiv.org/abs/2605.27490}, 
}

@inbook{doi:10.1137/1.9781611978971.47,
author = {Anupam Gupta and Amit Kumar and Debmalya Panigrahi and Zhaozi Wang},
title = {An Optimal Online Algorithm for Robust Flow Time Scheduling},
booktitle = {Proceedings of the 2026 Annual ACM-SIAM Symposium on Discrete Algorithms (SODA)},
chapter = {},
pages = {1214-1238},
doi = {10.1137/1.9781611978971.47},
URL = {https://epubs.siam.org/doi/abs/10.1137/1.9781611978971.47},
eprint = {https://epubs.siam.org/doi/pdf/10.1137/1.9781611978971.47}
}

@InProceedings{cabello_et_al:LIPIcs.ITCS.2026.31,
  author =	{Cabello, Sergio and Chan, Timothy M. and Giannopoulos, Panos},
  title =	{{Delaunay Triangulations with Predictions}},
  booktitle =	{17th Innovations in Theoretical Computer Science Conference (ITCS 2026)},
  pages =	{31:1--31:23},
  series =	{Leibniz International Proceedings in Informatics (LIPIcs)},
  ISBN =	{978-3-95977-410-9},
  ISSN =	{1868-8969},
  year =	{2026},
  volume =	{362},
  editor =	{Saraf, Shubhangi},
  publisher =	{Schloss Dagstuhl -- Leibniz-Zentrum f{\"u}r Informatik},
  address =	{Dagstuhl, Germany},
  URL =		{https://drops.dagstuhl.de/entities/document/10.4230/LIPIcs.ITCS.2026.31},
  URN =		{urn:nbn:de:0030-drops-253186},
  doi =		{10.4230/LIPIcs.ITCS.2026.31}
}

@InProceedings{10.1007/978-3-031-81396-2_3,
author="Cabello, Sergio
and Giannopoulos, Panos",
editor="Bie{\'{n}}kowski, Marcin
and Englert, Matthias",
title="Searching in Euclidean Spaces with Predictions",
booktitle="Approximation and Online Algorithms",
year="2025",
publisher="Springer Nature Switzerland",
address="Cham",
pages="31--45",
isbn="978-3-031-81396-2"
}

@InProceedings{pmlr-v267-daneshvaramoli25a,
  title = 	 {Near-Optimal Consistency-Robustness Trade-Offs for Learning-Augmented Online Knapsack Problems},
  author =       {Daneshvaramoli, Mohammadreza and Karisani, Helia and Lechowicz, Adam and Sun, Bo and Musco, Cameron N and Hajiesmaili, Mohammad},
  booktitle = 	 {Proceedings of the 42nd International Conference on Machine Learning},
  pages = 	 {12459--12489},
  year = 	 {2025},
  editor = 	 {Singh, Aarti and Fazel, Maryam and Hsu, Daniel and Lacoste-Julien, Simon and Berkenkamp, Felix and Maharaj, Tegan and Wagstaff, Kiri and Zhu, Jerry},
  volume = 	 {267},
  series = 	 {Proceedings of Machine Learning Research},
  month = 	 {13--19 Jul},
  publisher =    {PMLR},
  url = 	 {https://proceedings.mlr.press/v267/daneshvaramoli25a.html}
}

@inproceedings{
lechowicz2024time,
title={Time Fairness in Online Knapsack Problems},
author={Adam Lechowicz and Rik Sengupta and Bo Sun and Shahin Kamali and Mohammad Hajiesmaili},
booktitle={The Twelfth International Conference on Learning Representations},
year={2024},
url={https://openreview.net/forum?id=9kG7TwgLYu}
}

@misc{cui2026skirentaldistributionalpredictions,
      title={Ski Rental with Distributional Predictions of Unknown Quality}, 
      author={Qiming Cui and Michael Dinitz},
      year={2026},
      eprint={2602.21104},
      archivePrefix={arXiv},
      primaryClass={cs.LG},
      url={https://arxiv.org/abs/2602.21104}, 
}

@misc{angelopoulos2024contractschedulingdistributionalmultiple,
      title={Contract Scheduling with Distributional and Multiple Advice}, 
      author={Spyros Angelopoulos and Marcin Bienkowski and Christoph Dürr and Bertrand Simon},
      year={2024},
      eprint={2404.12485},
      archivePrefix={arXiv},
      primaryClass={cs.DS},
      url={https://arxiv.org/abs/2404.12485}, 
}

@misc{kim2026robustconsistentskirental,
      title={Robust and Consistent Ski Rental with Distributional Advice}, 
      author={Jihwan Kim and Chenglin Fan},
      year={2026},
      eprint={2603.29233},
      archivePrefix={arXiv},
      primaryClass={cs.LG},
      url={https://arxiv.org/abs/2603.29233}, 
}

@misc{kang2025learningaugmentedskirentaldiscrete,
      title={Learning-Augmented Ski Rental with Discrete Distributions: A Bayesian Approach}, 
      author={Bosun Kang and Hyejun Park and Chenglin Fan},
      year={2025},
      eprint={2512.07313},
      archivePrefix={arXiv},
      primaryClass={cs.LG},
      url={https://arxiv.org/abs/2512.07313}, 
}
\bibliographystyle{plainnat}

\newpage
\appendix
\section{Experimental Results} \label{experimental results}
\paragraph{Trials and evaluation.}
We conduct ten independent trials on each tree. In each trial,
we sample 50 targets uniformly without replacement from the true
support $A$ and use the same targets for all three methods and
all five replacement levels. We report the average number of
queries across the sampled targets and trials, giving 500 target
evaluations per plotted point. These averages estimate expected
query complexity under the true distribution. All trials and
outcomes are included.

\begin{figure}[h]
    \centering
    \includegraphics[width=0.48\linewidth]{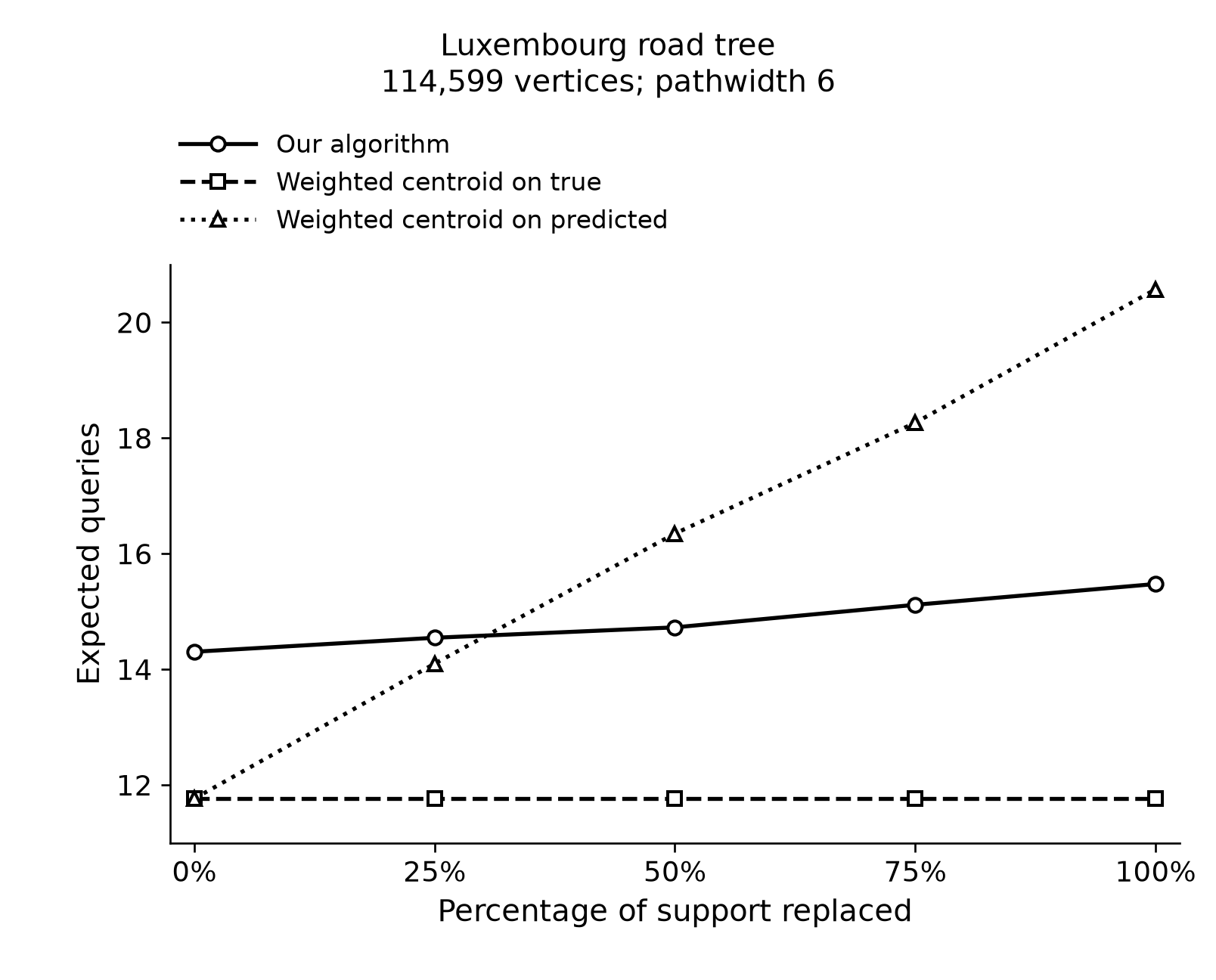}\hfill
    \includegraphics[width=0.48\linewidth]{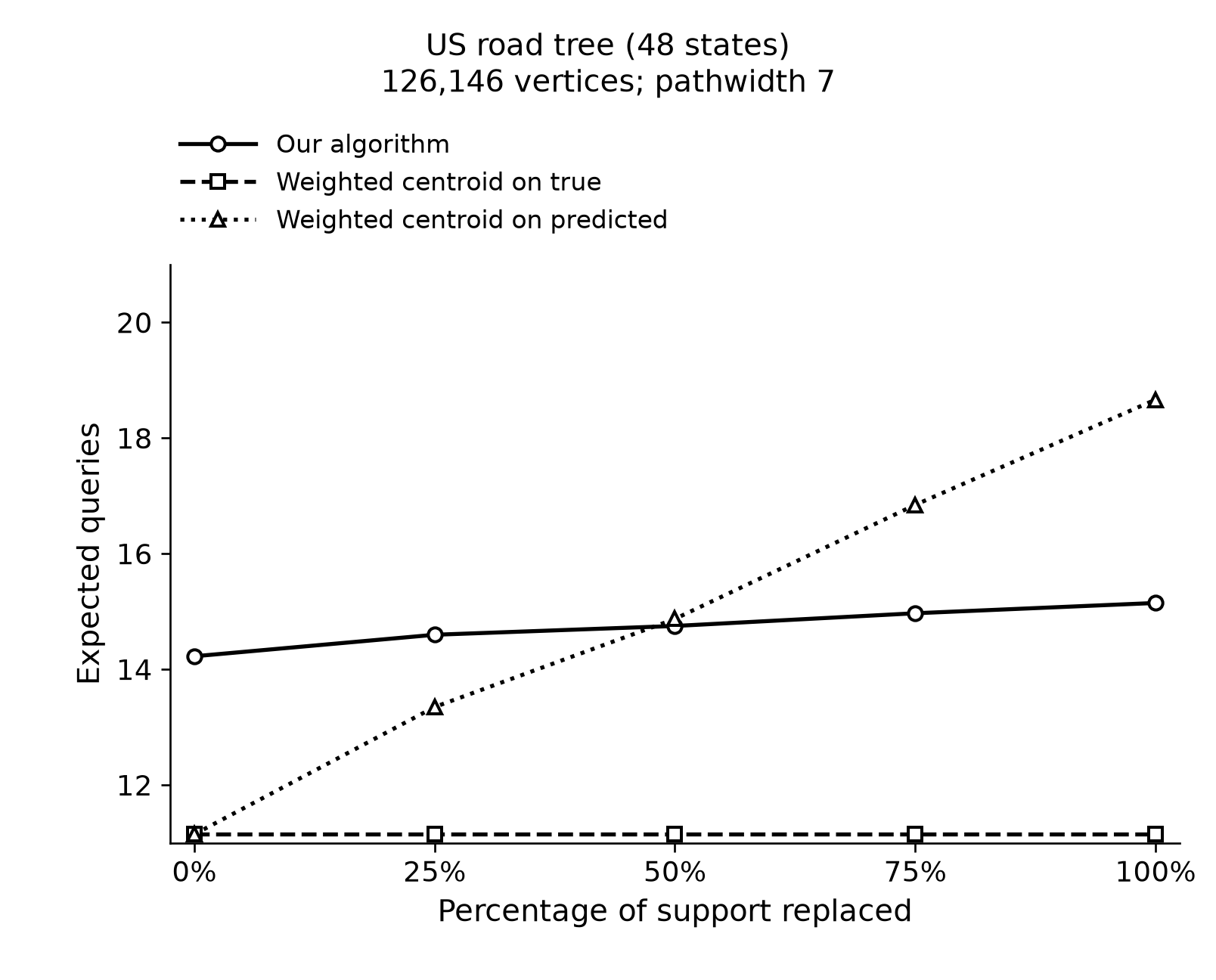}\\[2mm]
    \includegraphics[width=0.48\linewidth]{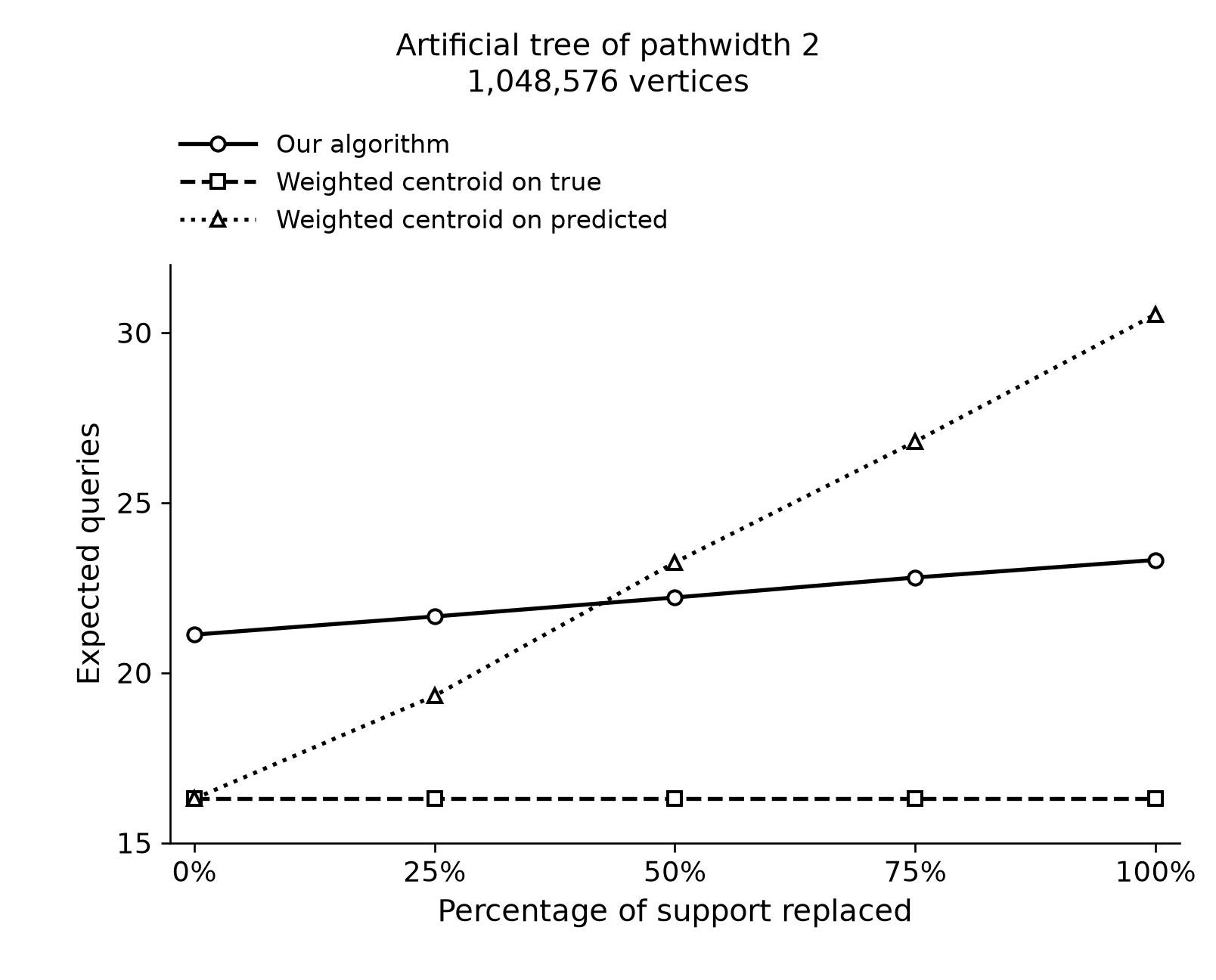}\hfill
    \includegraphics[width=0.48\linewidth]{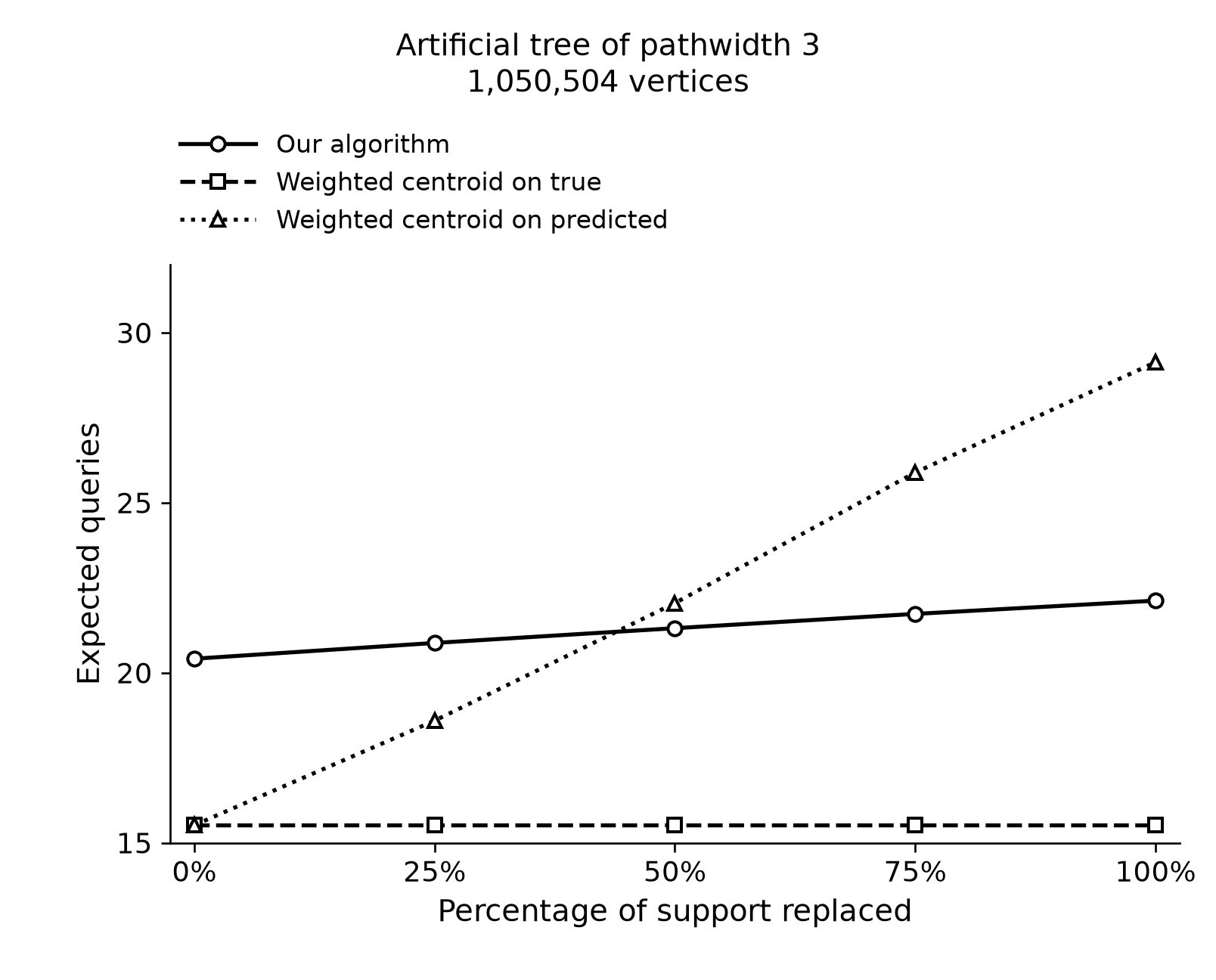}
    \caption{Expected query complexity versus percentage
    of support replaced. The top row shows the Luxembourg and US
    road trees; the bottom row shows the artificial trees of
    pathwidth 2 and 3.}
    \label{fig:tree-experiments}
\end{figure}

\paragraph{Performance under accurate predictions.}
Figure~\ref{fig:tree-experiments} shows a consistent tradeoff across
the four trees. With no replacement, the predicted and true
distributions coincide, so the two weighted-centroid baselines have
identical query counts. Our algorithm incurs additional cost from
its recovery procedures. On the Luxembourg and US road trees, its
mean query counts are $14.30$ and $14.23$, compared with $11.76$ and
$11.15$ for weighted centroid. On the artificial trees of pathwidth
2 and 3, the corresponding values are $21.12$ and $20.41$, compared
with $16.30$ and $15.52$. Across the four trees, this is an overhead
of approximately $22\%$--$32\%$ under perfect predictions.

The predicted-centroid baseline also has a lower mean query count
at $25\%$ replacement on every tree. These observations show that
our algorithm does not uniformly improve on prediction-guided
centroid search: its additional search procedures have a measurable
cost when predictions are accurate or only mildly perturbed.

\paragraph{Performance as support mismatch increases.}
As the replacement fraction increases, the mean query count of
the predicted-centroid baseline rises more sharply than that of
our algorithm. At the sampled levels of $50\%$, $75\%$, and $100\%$
replacement, our algorithm has a lower mean query count on all four
trees. The difference at $50\%$ replacement is small for the US
road tree, where the means are $14.75$ and $14.87$, so this point
should be interpreted as a close comparison.

At complete replacement, our algorithm uses $15.48$ queries on
average on the Luxembourg road tree and $15.15$ on the US road
tree, compared with $20.57$ and $18.66$ for the predicted-centroid
baseline. On the artificial trees of pathwidth 2 and 3, it uses
$23.32$ and $22.12$ queries, compared with $30.54$ and $29.13$.
These correspond to reductions of approximately $19\%$--$25\%$
relative to the restarting predicted-centroid baseline.

Over the full replacement range, our algorithm's mean query cost
increases by approximately $6.5\%$--$10.4\%$, whereas the
predicted-centroid baseline increases by approximately
$67.5\%$--$87.7\%$. The true-distribution reference remains constant
because its input distribution and evaluated targets do not change
within a trial.

Across both road networks and million-vertex artificial trees,
our algorithm maintains stable query costs as prediction error
increases. At complete support replacement, it uses approximately
$19\%$--$25\%$ fewer queries than the predicted-centroid
baseline. This consistent performance across all four trees
highlights the strength of our approach: substantial mismatch
between the predicted and true distributions leads to only a
limited increase in search cost.

\end{document}